\documentclass[12pt,reqno]{amsart}

\usepackage[samelinetheorem,notitle]{maherart}
\usepackage{pgfplots}
\usetikzlibrary{trees}
\usetikzlibrary{patterns}
\hypersetup{
    pdftitle =
        {Queue Design for Staking and Restaking},
    pdfauthor =
        {Mallesh M Pai and Max Resnick},
    pdfsubject=
        {Queue Design}
}

\usepackage{tabularx}
\usepackage{multirow}
\usepackage{longtable}

    \newcolumntype{L}{>{\raggedright\arraybackslash}X}

\usepackage{xcolor}

\usepackage{nicematrix,booktabs,caption}

\begin{document}
\raggedbottom

\author[Pai]{Mallesh M Pai}
\author[Resnick]{Max Resnick\\ Special Mechanisms Group, Consensys Inc}

\title{Centralization in Attester-Proposer Separation}

\date{\today}

\begin{abstract}
    We show that Execution Tickets and Execution Auctions dramatically increase centralization in the market for block proposals, even without multi-block MEV concerns. Previous analyses have insufficiently or incorrectly modeled the interaction between ahead-of-time auctions and just-in-time (JIT) auctions. We study a model where bidders compete in an execution auction ahead of time, and then the winner holds a JIT auction to resell the proposal rights when the slot arrives. During the execution auction, bidders only know the distribution of their valuations. Bidders then draw values from their distributions and compete in the JIT auction. We show that a bidder who wins the execution auction is substantially advantaged in the JIT auction since they can set a reserve price higher than their own realized value for the slot to increase their revenue. As a result, there is a strong centralizing force in the execution auction, which allows the ex-ante strongest bidder to win the execution auction every time, and similarly gives them the strongest incentive to buy up all the tickets. Similar results trivially apply if the resale market is imperfect, since that only reinforces the advantages of the ex-ante strong buyer. To reiterate, these results do not require the bidders to employ multi-block MEV strategies, although if they did, it would likely amplify the centralizing effects. 
\end{abstract}

    \maketitle

\newpage 
\section{Introduction}
Attester-Proposer Separation suggests splitting up attester duties (participating in consensus) from proposer duties (proposing the next block). Currently, validators perform both duties: all validators participate in consensus, while the validator responsible for proposing the next block is determined via in-protocol randomization. However, constructing a valuable block from various sources (public mempool, private orderflow, trading bots, etc) is a specialized activity, and therefore, currently, the right to propose a block is resold by the proposer in an out-of-protocol, real-time auction known as MEV-Boost. The main reason to separate these duties is to then capture the revenues generated by selling the right to propose a block in advance and in protocol.

Two instantiations have been proposed: Execution Tickets (ET),\footnote{\url{https://ethresear.ch/t/execution-tickets/17944}} and more recently, Execution Auctions (EA).\footnote{\url{https://ethresear.ch/t/execution-auctions-as-an-alternative-to-execution-tickets/19894}}  Both are fairly easy to describe. In the EA design, the right to be the proposer in slot n + d  is instead sold in an earlier slot n. In the ET design, lottery tickets are sold earlier by some mechanism and the winning lottery ticket for slot n+d is drawn using on-chain randomness.

Concerns about both proposals have already circulated, mostly centering around the fact that the block proposer being known in advance will enable multi-block MEV strategies.\footnote{See, e.g,. \url{https://x.com/_charlienoyes/status/1806186662327689441}.} In this short paper, we present a more basic concern: we are already in a setting where there are only a few competitive builders: 3 builders currently produce over 90\% of blocks. APS, implemented as either EA or ET, will lead to even more centralization at the builder level, further entrenching the best builder(s). In what follows, we formally state and prove the following result:
\begin{theorem*}[Informal]
    The ex-ante strongest builder is willing to bid higher than anyone else in an Execution Auction, and therefore wins every Execution Auction. Analogously, this builder places a higher valuation than anyone else for Execution Tickets and will buy up all of them. Further, this builder, in equilibrium, produces a larger fraction of blocks than they would under MEV-Boost, i.e., APS magnifies their inherent advantage over other builders. 
\end{theorem*}

We show this in the context of a game-theoretic model. The game proceeds in two stages: In the first stage, there is the ahead-of-time execution auction for the slot. Since this is an ahead-of-time auction, bidders do not know what their realized values for the slot will be at this time (e.g., they do not know what transactions will actually be in the mempool at the time of the slot, what arbitrage opportunities will exist, etc.). Instead, they only know the distributions of their values. Critically, some bidders may be stronger than others, in the sense that the distribution from which they may realize their value is superior to others.%
\footnote{We formalize this and show empirical evidence of this assumption below.} 

In the second stage, bidders learn their value for the slot (drawn from their distribution). The winner of the execution auction in the first stage may then hold a resale auction (of his choosing) to resell the right to propose a block in that slot.\footnote{The original proposals themselves identify this:
\begin{quote}
   \emph{A secondary market will most likely develop where an EA ticket winner can resell their proposer right before their turn to propose. Even if the protocol does not allow them to transfer that right, this can be easily done via an out-of-protocol gadget.}
\end{quote}}
Since this is a thin, oligopolistic market, we assume that the winner optimally exercises market power in their resale, formally, reselling to maximize their net profit/ revenue. This is achieved by setting a higher reservation price to resell the right than their own realized value for the block.\footnote{As we argue later, a similar result holds more generally if the resale market is imperfect.}

We solve for the subgame-perfect Nash equilibrium of this game, that is to say, bidders bid in the execution auction taking into account that a) if they own the rights to propose and have the option to either (optimally) resell this right, or to propose their own block, and, b) if they do not own the rights then they can only bid in the resale auction run by the original owner of the rights, and receive surplus, and prove the result shown above. 

It is useful, at this stage, to briefly discuss our assumptions:
\begin{enumerate}
\item \emph{In advance of the actual slot, block builders draw values from different distributions}: To see why this is reasonable, note that currently there are 2 builders who win almost all of the real-time MEV-Boost auctions. Further, the top builder wins roughly half of the auctions, the second wins roughly four-tenths, and a long tail of builders share the remaining tenth.\footnote{Source: \url{www.mevboost.pics}.} 

Given that MEV-Boost uses a standard English auction format where it is a dominant strategy to bid up to one's value, to a first approximation, it is reasonable to conclude that these builders are drawing values from different distributions. 

\item \emph{The winner of the EA or owner of the winning lottery ticket in ET may try to resell the right}: As we pointed out above, this is explicitly considered in the initial proposals. 

\item \emph{The winner of the EA optimally exercises market power in the resale market}: we make this assumption for analytical convenience, we show by means of a numerical example that our main result also applies if they just apply a suboptimal markup (e.g. sell the right for, say, 10\% more than the value of their own best block). 
\end{enumerate}

We discuss the implications of our findings in Section \ref{sec:implications}
\section{Model}\label{sec:model}

Formally, consider the following extensive-form game:

\paragraph{Players}
The players are 2 builders and any number of non-builders. 

\paragraph{Timing} The game proceeds over two periods: 
\begin{enumerate}
    \item In Period $1$ (Slot n), the execution auction is held, which sells the right to propose the block for Slot n+d. We will model the execution auction as a standard English (ascending) auction, or equivalently a sealed-bid second-price auction.  
    \item In the Period $2$ (Slot n+d), the winner in the previous period can either use the right and propose a block, or resell it. 
    
    If the winner in Period $1$ chooses to resell the right, the builders bid in an auction of the current owner's choosing, and the revenue accrues to that player.
\end{enumerate}

\paragraph{Information and Payoffs} 
In period $2$, players privately learn their value for actually proposing a block (this is the sum of their values from sequencing orders in the public mempool, private orderflow, arbitrage opportunities, etc.). Each builder $i$ has a private value $v_i$ which is a draw from a distribution with CDF $F_i$ and continuous density $f_i$ on support normalized to $[0,1]$. We assume that the these have non-decreasing hazard rates as is standard in mechanism design (see, e.g., \cite{myerson1981optimal}). Finally, we assume that $F_1 \succ F_2$ where $\succ$ denotes ordering in the Hazard rate ordering sense.%
\footnote{This implies but is stronger than the assumption that $F_1$ first order stochastically dominates $F_2$, see, e.g., \cite{shaked2007stochastic}.}
Any non-builder has a 0 value to actually propose a block. 

In period $1$, the realized values are unknown, but the distributions are common knowledge among the players.   

All players are risk-neutral and expected utility maximizers and have quasilinear utilities. Payoffs are straightforward: the eventual owner of the proposal right realizes their value $v_i$ from proposing the block. The winner in period $1$ gets the revenues from resale. There is no time discounting. 

\paragraph{Strategies and Solution Concept} In period $2$, the winner in period $1$ resells the object using the optimal auction \citep{myerson1981optimal}. In period $1$, we consider the standard truthful equilibrium of the execution auction where each player bids their value. Players are forward-looking, i.e. we employ the appropriate subgame-perfection concept. In what follows, we refer to an equilibrium satisfying these refinements as simply \emph{equilibrium}. 

We summarize the game in the following picture:
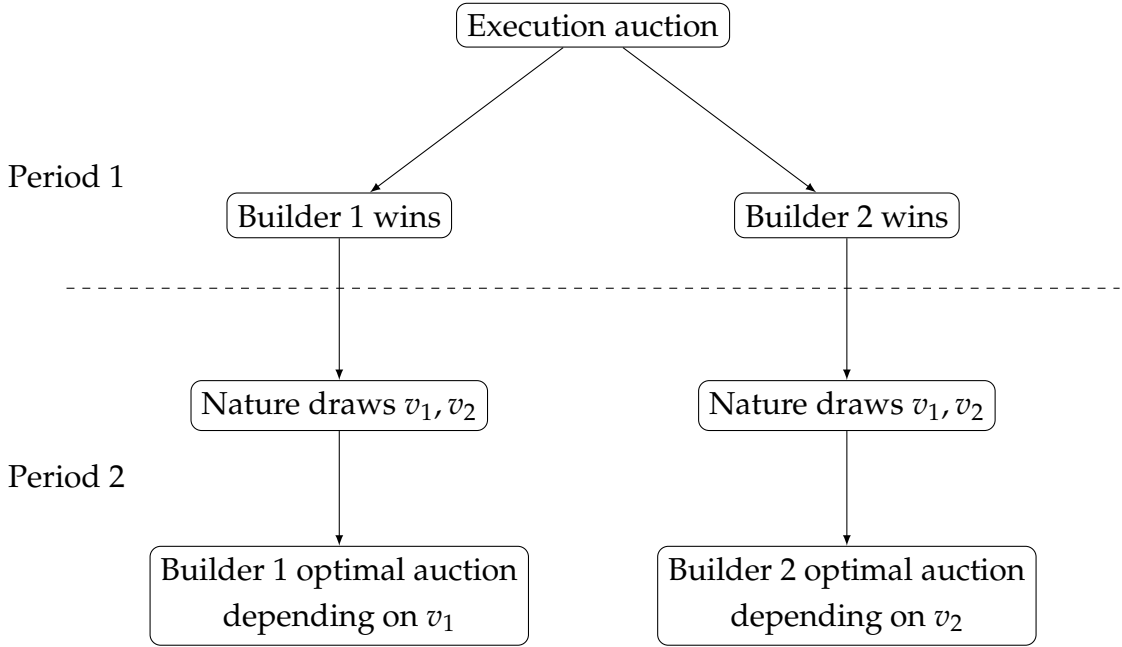
\begin{center}
\begin{tikzpicture}[
  edge from parent/.style={draw, -latex},
  sibling distance=16em, 
  level distance=6em,
  every node/.style={draw, rounded corners, align=center}
  ]
  \node {Execution auction}
    child {node {Builder 1 wins}
      child {node {Nature draws \(v_1, v_2\)}
        child {node {Builder 1 optimal auction\\ depending on $v_1$ }}}
    }
    child {node {Builder 2 wins}
      child {node {Nature draws \(v_1, v_2\)}
        child {node {Builder 2 optimal auction\\ depending on $v_2$ }}}
    };

  \draw[dashed] (-7, -3.5) -- (7, -3.5);
\node[draw=none, fill=none]  at (-7, -2) {Period 1};
  \node[draw=none, fill=none]  at (-7, -6) {Period 2};

\end{tikzpicture}

\end{center}
\section{Results}

We are now in a position to describe our main theorem.
\begin{theorem}\label{thm:aps}
    There is a unique equilibrium of the execution auction game. In this equilibrium, Builder $1$ always wins the auction in Period $1$. 
\end{theorem}

\begin{proof}
We solve this game via backward induction below:

\noindent \textbf{Period 2}: There are two cases in this period depending on who owns the rights to propose from the period 1 auction.
    \begin{enumerate}
        \item \textit{A non-builder owns the right}: In this case, they put the right up for resale. The 2 builders bid, and the resulting revenue is the second highest out of $v_1, v_2$. Let us denote this as $p_0 = \mathbb{E}[v^{1:2}]$. Let us call the expected surplus of the builders in this case $s^1_0$ and $s^2_0$.
        \item \textit{Builder $i$ owns the right}:  Say e.g., Builder 1 owns the right. In this case, Builder $1$ will offer to resell the right to the other builder. Since there is only one other builder (no one else will bid since their value is normalized to $0$), this is optimally in the form of a take-it-or-leave-it-offer, which depends on their own realized value for the block. Formally, if builder $1$ owns the right, and realizes a value of $v_1$, then, by \cite{myerson1981optimal} the profit-maximizing price to offer the other builder is the value $v_2^*(v_1)$ which solves: 
        \begin{align}
            v_2^* - \frac{1-F_2(v_2^*)}{f_2(v_2^*)} = v_1.
        \end{align}
        The analogous formula in the case builder $2$  is the initial owner defines their optimal offer price, $v_1^*(v_2).$
        
        If the other builder declines this offer, then Builder $1$ simply proposes their own built block and realizes the value of that block, $v_1$.

        Let us define the net profit of builder $1$ when they are the owner of the right by $p_1$ and the surplus of the other builder $2$ by $s^2_1$. Similarly, when builder $2$ wins we denote this by $p_2$ and the surplus of the other builder $1$ by $s^1_2.$
    \end{enumerate}

\noindent \textbf{Period 1}: Parties bid in the Execution auction. Note that this is a second-price auction, and we consider the standard truthful equilibrium where everyone bids their values. Firstly, note that any non-builder will bid $p_0$, since they get $0$ if they lose the auction. Conversely, builder $1$ stands to gain $p_1$ if they win; while if they lose they will still make either $s^1_2$ or $s^1_0$. Therefore their maximum willingness to pay is $p_1 - \min (s^1_2, s^1_0)$. The maximum willingness to pay of builder $2$ can be computed similarly.

\medskip 
Lemma \ref{lem:1} shows that under our maintained assumptions, $p_1 - s^1_2 > p_2 - s^2_1 > p_0$. Therefore, builder $1$ has the highest willingness to pay in the first-period auction and wins under the standard truthful equilibrium in weakly dominant strategies. 
\end{proof}

Of course, it isn't just the winner of the first-period auction that concerns us. The following Corollary shows that an execution auction would result in the ex-ante strongest builder also then proposing blocks more often than they would have under MEV-Boost.  This is an additional concern--- if private order flow were to exclusively contract with a single builder, as is rumored to be in some cases (for example certain Telegram bots send their flow to a single builder they have contracted with), they are more likely to choose the dominant builder for quality of execution reasons etc. 

\begin{corollary}
    In equilibrium, Builder $1$ proposes the block with higher probability than if the right to propose had been sold in Period $2$ (e.g., via MEV-Boost). 
\end{corollary}
\begin{proof}
    This follows straightforwardly from the theorem (Builder 1 wins the execution auction). Further, we showed that when builder $1$ owns the right to propose, and has a value of $v_1$, it will offer this right to builder $2$ for a take-it-or-leave it price of $v_2^*(v_1) > v_1$, so builder $2$ will only accept and propose the block when $v_2 \geq v_2^*(v_1)$. Note that since MEV-Boost is an English auction, we have under the usual equilibrium in truthful (weakly dominant) strategies, Builder $2$ wins the right auction whenever $v_2 \geq v_1$. The former event is clearly a subset of the latter, so the result follows. 
\end{proof}

This corollary straightforwardly implies that the expected surplus (profit) of Builder 2 is lower than under the MEV-Boost auction.%
\footnote{Formally this is a trivial consequence of revenue equivalence, see, e.g.,\cite{krishna2009auction}.}
This has additional long-term consequences for the builder market: for example, smaller/ fringe builders may find it even harder to sustain their presence than currently, resulting in further centralization. 

\section{Numerical Examples}
The theorem above may be more concrete via some straightforward numerical examples. The first is analytic. The second considers the case that there are more than 2 builders. 
\subsection*{Example 1}
To see this by example, consider the simple case where buyer $1$ has a value distributed $U[0,5/4]$ while buyer $2$ has a value distributed $U[0,1]$. Note that in this case, buyer $1$ wins a JIT auction with a probability of $3/5$ while buyer $2$ wins with a probability of  $2/5$.\footnote{These probabilities are proportional to the relative frequencies with which the top two builders win in the current MEV-Boost auction.} 

Now let's do some simple calculations. Firstly, note that a JIT auction in this case will achieve a revenue of $$\underbrace{\frac15}_{\substack{\text{Probability of builder 1}\\ \text{having value in $[1,5/4]$}}} \times \underbrace{\frac12}_{\text{Expected revenue}} + \underbrace{\frac45}_{\substack{\text{Complementary}\\\text{probability}}} \times \underbrace{\frac13}_{\substack{\text{Expected revenue}\\\text{of second price}}} = \frac{11}{30}.$$
Straightforward calculations show that builder $1$'s expected surplus in this auction is $\frac{4}{30} + \frac{1}{32}$ while builder $2's$ expected surplus in this auction is $\frac{4}{30}.$

Next, suppose the winner in the EA in period $1$ is builder $1$. In this case,  in period $2$ with probability $1/5$, builder $1$ has a value in $[1,5/4]$ and keeps the right to themselves. Conversely, with the remaining probability, builder $1$ has a value in $v_1 \in [0,1]$, and optimally offers the right to builder $2$ at a take-it-or-leave-it price of $\frac{1+v_1}{2}$ to maximize the expected profit (i.e., either getting a profit equal to the price if buyer $2$ accepts, and otherwise getting a profit equal to $v_1$). Some straightforward algebra implies that the total expected profit to builder $1$ in this case is
$$p_1= \frac{33}{40}.$$
Buyer $2$'s expected surplus is $s^2_1= \frac{1}{30}.$

In the reverse case where Builder $2$ wins the EA, and then realizes a value $v_2 \in [0,1]$, they optimally offer the right to builder $1$ at a take-it-or-leave-it price of $\frac{5/4+v_2}{2}$ to maximize the expected profit. Their total expected profit is therefore $p_2 = \frac{604}{15\times 64} \approx 0.629.$ Buyer $1's$ expected surplus in this case is $\frac{124}{1870} \approx \frac{1}{15}.$

By simple backward induction, therefore, in the first period auction, buyer $1$ is willing to pay $\frac{33}{40} - \frac{1}{15} \approx 0.758 $ which is larger than buyer $2's$ WTP which is $0.629 - \frac{1}{30} \approx 0.596$, both of which are larger than any non-builder-buyer's willingness to pay!

In short, in line with our Theorem, builder $1$ will win any execution auction and be willing to purchase any ET ahead of time. 

\subsection*{Example 2}
Of course, in practice there are more than $2$ builders. At the time of this writing, there are 3 major builders, and a long tail of smaller builders. The current market shares are roughly $\approx 50\%$ for the largest, $\approx 40\%$ for the second, and $\approx 7\%$ for the third, with the remainder split across several smaller builders.

To validate our model in this richer setting, let us instead consider a setting with 3 builders, with market shares of $50\%$, $40\%$ and $10 \%$ respectively.

We assume that each builder $i$ has a Lognormal distribution with parameters $(\mu_i, 1)$.%
\footnote{We choose the Lognormal distribution because among ``standard'' distributions, this distribution most closely matches observed bids in the MEV-Boost auction.}
Assume that $\mu_1 > \mu_2 > \mu_3$, in particular, $\mu_1 = 2.18,$, $\mu_2= 1.99$ and $\mu_3=1$. For these values, the probability that builder $i$ has the highest value among independent draws from these distributions is approximately $0.5, 0.4, 0.1$, i.e. roughly the outcome of MEV-Boost as described above. 

When buyers are ex-ante heterogeneous, the optimal auction is also discriminatory \citep{myerson1981optimal}. To simplify, we instead suppose that the winning builder runs a second-price auction among the other builders but adds a reserve price that is a marked-up version of its own realized value. Optimal revenues are achieved at a markup of $3$ (i.e., the winning builder in the execution auction offers to resell to the other two in a second price auction with a reservation price that is three times its own realized value). Our findings are summarized in the diagram below. 

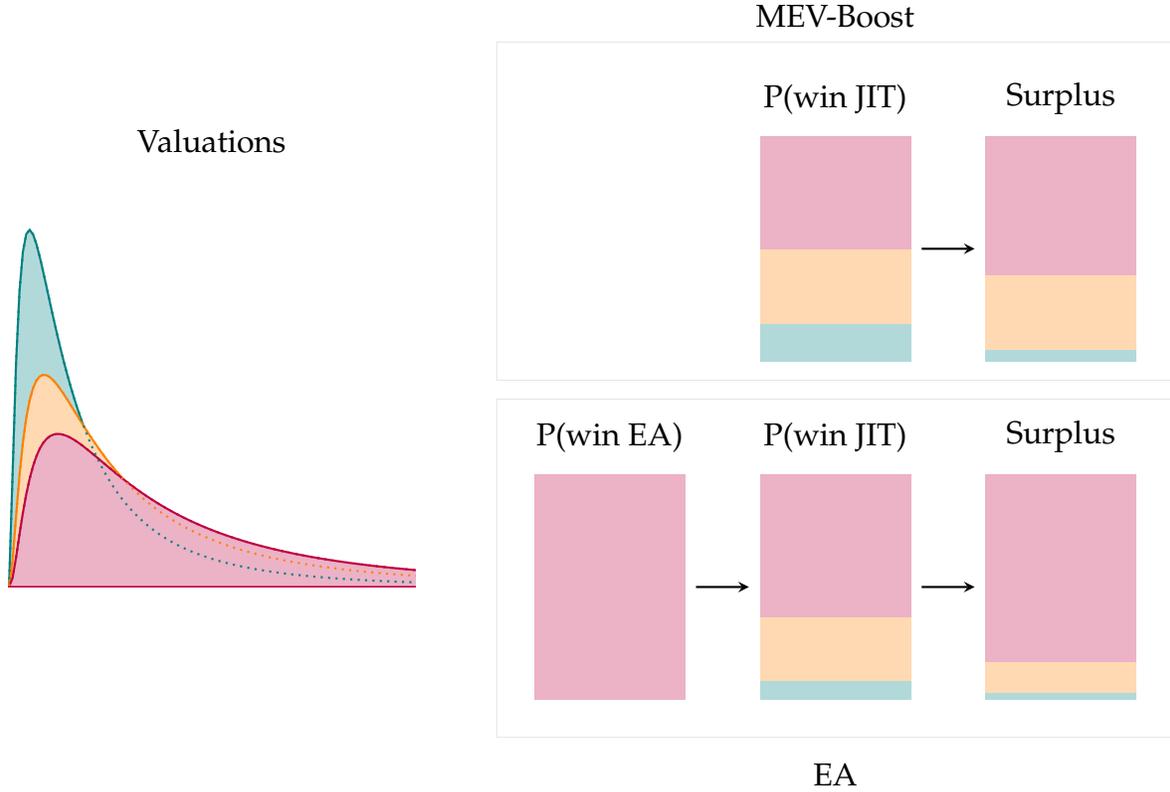
\begin{tikzpicture}[
    >=stealth,
    arrow/.style={->, thick},
    box/.style={draw, rectangle, minimum width=2cm, minimum height=1cm},
]

\begin{axis}[
    axis lines = none,
    ticks = none,
    xmin=0, xmax=3,
    ymin=0, ymax=1.75,
    width=7cm,
    height=7cm,
    samples=200,
    domain=0.01:5,
    legend pos=north east,
    legend cell align={left},
    title = Valuations
]

\addplot[teal, fill=teal!30, thick] {1/(x*sqrt(2*pi)) * exp(-((ln(x)+0.8473)^2 / 2))} \closedcycle;

\addplot[orange, fill=orange!30, thick] {1/(x*sqrt(2*pi)) * exp(-((ln(x)+0.3266)^2 / 2))} \closedcycle;

\addplot[purple, fill=purple!30, thick] {1/(x*sqrt(2*pi)) * exp(-(ln(x))^2 / 2)} \closedcycle;

\addplot[purple, dotted, thick] {1/(x*sqrt(2*pi)) * exp(-(ln(x))^2 / 2)};
\addplot[orange, dotted, thick] {1/(x*sqrt(2*pi)) * exp(-((ln(x)+0.3266)^2 / 2))};
\addplot[teal, dotted, thick] {1/(x*sqrt(2*pi)) * exp(-((ln(x)+0.8473)^2 / 2))};

\end{axis}



\node (win_prob) at (11,4.5) {
    \begin{tikzpicture}
    \fill[teal!30] (0,0) rectangle (2,.5);
    \fill[orange!30] (0,.5) rectangle (2,1.5);
    \fill[purple!30] (0,1.5) rectangle (2,3);
    \end{tikzpicture}
};\node (win_prob2) at (11,0) {
        \begin{tikzpicture}
    \fill[teal!30] (0,0) rectangle (2,.25);
    \fill[orange!30] (0,.25) rectangle (2,1.1);
    \fill[purple!30] (0,1.1) rectangle (2,3);
    \end{tikzpicture}
};

\node (win_prob3) at (8,0) {
    \begin{tikzpicture}
    \fill[teal!30] (0,0) rectangle (2,.5);
    \fill[orange!30] (0,0) rectangle (2,0);
    \fill[purple!30] (0,0) rectangle (2,3);
    \end{tikzpicture}
};

\node[above] at (win_prob.north) {P(win JIT)};
\node[above] at (win_prob2.north) {P(win JIT)};
\node[above] at (win_prob3.north) {P(win EA)};

\node (surplus) at (14,4.5) {
        \begin{tikzpicture}
    \fill[teal!30] (0,0) rectangle (2,0.161);
    \fill[orange!30] (0,0.161) rectangle (2,1.15);
    \fill[purple!30] (0,1.15) rectangle (2,3);
    \end{tikzpicture}
};

\node (surplus2) at (14,0) {
    \begin{tikzpicture}
    \fill[teal!30] (0,0) rectangle (2,0.1);
    \fill[orange!30] (0,0.1) rectangle (2,.5);
    \fill[purple!30] (0,.5) rectangle (2,3);
    \end{tikzpicture}
};
\node[above] at (surplus.north) {Surplus};
\node[above] at (surplus2.north) {Surplus};

\draw[arrow] (win_prob) -- (surplus);
\draw[arrow] (win_prob2) -- (surplus2);
\draw[arrow] (win_prob3) -- (win_prob2);

\draw [rectangle,draw = lightgray!40] (6.5,2.75) rectangle (15.5,7.25);
\node at (11,7.6) {MEV-Boost};

\draw [rectangle,draw = lightgray!40] (6.5,-2) rectangle (15.5,2.5);
\node at (11,-2.5) {EA};
\end{tikzpicture}

In text: Numerical simulations then show that when Builder 1 wins the execution auction, it achieves a total expected value of 16.497 in period $2$, while builder $2$ and builder $3$'s surpluses are 2.645 and 0.273 respectively. Similarly, if Builder 2 wins the execution auction its expected value in period $2$ is 14.63, while builders $1$ and $3$ have expected surpluses of 4.338 and 0.307 respectively. Finally if Builder 3 wins the execution auction, its expected value is 9.576 while builders $1$ and $2's$ surpluses are 6.071 and 4.157 respectively. 

By backward induction, therefore, in the execution auction in period $1$, builder $1$'s willingness to pay for the right to is $16.497 - 4.338 = 12.159$, builder $2$'s willingness to pay is $14.63-2.645 = 11.985$, while builder $3$'s is $<10$. Therefore, builder $1$ wins the execution auction. 

Finally, our simulations also show that in the resulting equilibrium, the block is actually proposed by Builder $1$ with probability $0.795$, Builder $2$ with probability $0.171$, and builder $3$ with probability $0.034$. Compared to the MEV-Boost shares of $(0.5,0.4,0.1)$ the execution auction essentially decimates builder $3$'s share and halves builder $2$'s share, while entrenching builder $1$.

\section{Implications and Conclusions} \label{sec:implications}

In conclusion, our model strongly suggests Execution Auctions will result in even further centralization of the builder market. The execution tickets proposals do not formally specify how the underlying tickets will be priced/ sold. Nevertheless, our results imply that the ex-ante strongest builder will place a higher willingness to pay on a ticket than the other parties and will, therefore, be at a substantial advantage. Similar results trivially apply if the resale market is imperfect, since that only reinforces the advantages of the ex-ante strong buyer. Similar results trivially apply if the resale market is imperfect. For example if there is no resale market, it is obvious that the ex-ante strongest buyer wins the auction. 

The biggest concern, in our minds is what such proposals might do to the health of competition in the overall builder market. We are already in a setting where two builders produce most blocks ($\sim 90\%$ as of this writing), and three builders produce almost all. In the long run, such proposals will, in our opinion, further restrict builder entry and/or lead to the exit of smaller builders currently in the market. 

To our minds, as a result, these proposals are inherently flawed and unworkable. Selling a ``right'' ex-ante when different parties are highly heterogeneous inherently advantages the party that is ex-ante stronger, and removes the role of randomness (e.g. randomness in realized flows to each builder etc.) in creating a more competitive builder market. 

This leaves us with the question of what should Ethereum do? The current out-of-protocol solution is meant to be temporary, and there are several technical implementation issues in conducting such a JIT auction in protocol (``ePBS''). In our opinion, this dichotomy is a false one, and instead of focusing attention on bringing current MEV payments into protocol, Ethereum research should focus its attention on other proposals that decrease the central role of the builder and/or the reliance on a monolithic proposer. These include our preferred solution, Multiple Concurrent Proposers \citep{fox2023censorship}, and others such as Inclusion Lists.\footnote{See, e.g., \url{https://ethresear.ch/t/fork-choice-enforced-inclusion-lists-focil-a-simple-committee-based-inclusion-list-proposal/19870}.} Even more straightforward design choices like speeding up the block production rate (currently a block is produced every 12 seconds) should 
greatly reduce the amount MEV and may decrease the ex-ante differences between buidlers. 

\bibliographystyle{econometrica}
\bibliography{refs}

\appendix

\begin{lemma}\label{lem:1}
    Under the maintained assumptions of Section \ref{sec:model}, we have that $p_1 - s^1_2 > p_2 - s^2_1 > p_0$. 
\end{lemma}

\begin{proof}
    Note that: 
\begin{align*}
    &p_1 = \mathbb{E}_{v_1}\left[v_2^*(v_1) (1-F_2(v_2^*(v_1))) + v_1 F_2(v_2^*(v_1))\right],\\
    &p_2 = \mathbb{E}_{v_2}\left[v_1^*(v_2) (1-F_1(v_1^*(v_2))) + v_2 F_1(v_1^*(v_2))\right],\\
    &s^1_2 = \mathbb{E}_{v_1, v_2}[(v_1 - v_1^*(v_2))^+],\\
    &s^2_1 = \mathbb{E}_{v_1, v_2}[(v_2 - v_2^*(v_1))^+].
\end{align*}
Here $(x)^+ \equiv \max \{x, 0\}$.

Therefore:
\begin{align*}
   & p_1 - s^1_2 > p_2 - s^2_1\\
   \iff & \mathbb{E}_{v_1}\left[v_2^*(v_1) (1-F_2(v_2^*(v_1))) + v_1 F_2(v_2^*(v_1))\right] - \mathbb{E}_{v_1, v_2}[(v_1 - v_1^*(v_2))^+] \\
   &\qquad > \mathbb{E}_{v_2}\left[v_1^*(v_2) (1-F_1(v_1^*(v_2))) + v_2 F_1(v_1^*(v_2))\right] - \mathbb{E}_{v_1, v_2}[(v_2 - v_2^*(v_1))^+],\\
   \iff & \mathbb{E}_{v_1}\left[ v_1 F_2(v_2^*(v_1))\right] - \mathbb{E}_{v_1, v_2}[v_1 \mathbbm{1}_{v_1 > v_1^*(v_2)} ] > \mathbb{E}_{v_2}\left[ v_2 F_1(v_1^*(v_2))\right] - \mathbb{E}_{v_1, v_2}[(v_2 \mathbbm{1}_{v_2 > v_2^*(v_1)}],\\
   \iff & \mathbb{E}_{v_1}\left[ v_1 F_2(v_2^*(v_1))\right] + \mathbb{E}_{v_1, v_2}[(v_2 \mathbbm{1}_{v_2 > v_2^*(v_1)}]  > \mathbb{E}_{v_2}\left[ v_2 F_1(v_1^*(v_2))\right] + \mathbb{E}_{v_1, v_2}[v_1 \mathbbm{1}_{v_1 > v_1^*(v_2)} ],\\
    \iff & \int_0^1 \int_0^{v_2^*(v_1)} v_1 f_2(v_2) dv_2 f_1 (v_1) dv_1 + \int_0^1 \int_{v_2^*(v_1)}^1 v_2 f_2(v_2) dv_2 f_1 (v_1) dv_1 \\
    & \qquad > \int_0^1 \int_0^{v_1^*(v_2)} v_2 f_1(v_1) dv_1 f_2 (v_2) dv_2 + \int_0^1 \int_{v_1^*(v_2)}^1 v_1 f_1(v_1) dv_1 f_2 (v_2) dv_2. 
\end{align*}
Now define two functions, $g_1, g_2: [0,1]^2 \to [0,1].$  as:
\begin{align*}
    &g_1(v_1,v_2) = \begin{cases}
        v_2 & \textrm{ if } v_2 > v_2^*(v_1),\\
        v_1 & \textrm{o.w.}
    \end{cases}\\
\intertext{and, similarly,}
& g_2(v_1, v_2) = \begin{cases}
    v_1 & \textrm{ if } v_1 > v_1^*(v_2),\\
    v_2 & \textrm{o.w.}
\end{cases}
\end{align*}
We can therefore rewrite the previous inequality as:
\begin{align*}
    &\int_0^1 \int_0^1 g_1(v_1,v_2) f_1(v_1) dv_1 f_2 (v_2) dv_2 > \int_0^1 \int_0^1 g_2(v_1,v_2) f_1(v_1) dv_1 f_2 (v_2) dv_2\\
    \iff & \int_0^1 \int_0^1 \left( g_1(v_1,v_2) - g_2 (v_1,v_2)\right) f_1(v_1) dv_1 f_2 (v_2) dv_2 >0
\end{align*}
    Now, note that:
    \begin{align*}
        g_1(v_1, v_2) - g_2(v_1, v_2) = \begin{cases}
            0 & \textrm{ if } v_2 > v_2^*(v_1) \vee v_1 > v_1^*(v_2),\\
            v_1-v_2 & \textrm{ if } v_2 \leq v_2^*(v_1) \wedge v_1 \leq v_1^*(v_2).
        \end{cases}
    \end{align*}
    The result now follows since our assumptions of monotone hazard rate, we have that $v^*_1(\cdot), v^*_2(\cdot)$ are both strictly increasing. Further from the assumption that that $F_1$ dominates $F_2$ in the hazard rate order, we have that for any $x \in [0,1]$, $v^*_1(x) \geq v^*_2(x)$. Finally, from the definition of hazard rate we have that $v^*_1(1)  = v^*_2(1) = 1$. To see , consider the function $g_3$ defined as:
    \begin{align*}
       g_3 (v_1,v_2) =  \begin{cases}
            0 & \textrm{ if } v_2 > v_2^*(v_1) \vee v_1 > v_1^*(v_2),\\
            v_1-v_2 & \textrm{ if } v_2 \leq v_2^*(v_1) \wedge v_1 \leq v_2^*(v_2).
        \end{cases}
    \end{align*}
    Note that by symmetry of the function $g_3$ around the diagonal $v_1 = v_2$, we have that $\int_0^1 \int_0^1 g_3(v_1, v_2) f_1(v_1) dv_1 f_2 (v_2) dv_2 =0$, Finally note that that 
    \begin{align*}
       g_1 - g_2 - g_3 (v_1,v_2) =  \begin{cases}
            v_1-v_2 & \textrm{ if }  v_1 > v_2^*(v_1) \wedge v_1 < v_2^*(v_1),\\
            0 & \textrm{ o.w. }
        \end{cases}
    \end{align*}
    Finally note that this function is non-negative on its entire domain. Therefore we have $ \int_0^1 \int_0^1 \left( g_1(v_1,v_2) - g_2 (v_1,v_2)\right) f_1(v_1) dv_1 f_2 (v_2) dv_2 >0$ as desired. 

    An analogous argument shows that $p_1 - s^1_2 >p_0$ and concludes the proof. 

\end{proof}
\end{document}